\documentclass[letterpaper,11pt]{article}

\usepackage{amsthm}
\usepackage{amssymb}
\usepackage{amsmath}
\usepackage{amsfonts}
\usepackage{amstext}

\usepackage[sort,compress,numbers]{natbib}

\usepackage{fullpage}

\usepackage{color}

\usepackage[T1]{fontenc}
\usepackage[utopia]{mathdesign}

\newcommand{\tinyspace}{\mspace{1mu}}

\newcommand*{\bra}[1]{\langle #1|}
\newcommand*{\ket}[1]{|#1\rangle}
\newcommand{\ketbra}[1]{\ket{#1}\bra{#1}}
\newcommand{\braket}[2]{\langle #1|#2\rangle}

\DeclareRobustCommand\openone{\leavevmode\hbox{\upshape\small1\normalsize\kern-.33em1}}%
\newcommand{\tensor}{\otimes}
\newcommand{\tprod}{\otimes}
\newcommand{\norm}[1]{\left\lVert\tinyspace#1\tinyspace\right\rVert}
\newcommand{\tnorm}[1]{\norm{#1}_{\mathrm{tr}}}
\newcommand{\dnorm}[1]{\norm{#1}_{\diamond}}
\newcommand{\opnorm}[1]{\norm{#1}_{\infty}}
\newcommand{\smallnorm}[1]{\bigl\lVert\tinyspace#1\tinyspace\bigr\rVert}

\newcommand{\smalldnorm}[1]{\smallnorm{#1}_{\diamond}}

\newcommand{\abs}[1]{\left\lvert\tinyspace #1 \tinyspace\right\rvert}

\newcommand{\dm}[1]{\dim \mathcal{#1}}

\newcommand{\linear}[1]{\mathbf{L}(\mathcal{#1})}
\newcommand{\density}[1]{\mathbf{D}(\mathcal{#1})}

\newcommand{\transform}[1]{\mathbf{T}(\mathcal{#1})}
\newcommand{\id}{\openone}
\newcommand{\identity}[1]{\id_{\mathcal{#1}}}

\newcommand{\tidentity}[1]{\id_{\mathcal{#1}}}

\newcommand{\tr}{\operatorname{tr}}

\newcommand{\ptr}[1]{\tr_\mathcal{#1}}
\newcommand{\ceil}[1]{\left\lceil #1 \right\rceil}

\newcommand{\prob}[1]{\textup{\textsc{#1}}}
\newcommand{\class}[1]{\textup{\textsf{#1}}}

\newtheorem{theorem}{Theorem}
\newtheorem{lemma}[theorem]{Lemma}
\newtheorem{proposition}[theorem]{Proposition}

\theoremstyle{definition}
\newtheorem{defn}[theorem]{Definition}
\newtheorem{proto}[theorem]{Protocol}
\newtheorem{problem}[theorem]{Problem}

\usepackage{hyperref}

\title{Testing quantum circuits and detecting insecure encryption}

\author{Bill Rosgen \\
  Centre for Quantum Technologies \\
  National University of Singapore \\
  \texttt{bill.rosgen@nus.edu.sg}}

\date{August 4, 2011}

\begin{document}

\maketitle

\begin{abstract}
  We show that computational problem of testing the behaviour of quantum circuits is hard for the class of problems known as~\class{QMA} that can be verified efficiently with a quantum computer.
  This result is a generalization of the techniques previously used to prove the hardness of other problem on quantum circuits.
 We use this result to show the~\class{QMA}-hardness of a weak version of the problem of detecting the insecurity of a symmetric-key quantum encryption system, or alternately the problem of determining when a quantum channel is \emph{not} private.
  We also give a \class{QMA} protocol for the problem of detecting insecure encryption to show that it is \class{QMA}-complete.
\end{abstract}


\section{Introduction}

Testing the behaviour of a computational system is a problem central to the study of quantum computing.
This is the problem faced by an experimentalist who has implemented a quantum computation and wants to check that the implementation behaves (approximately) correctly on all input states.
An efficient solution to this problem would allow for the verification that a circuit provided by an untrusted party correctly implements some desired operation.
Unfortunately we show in a general model that even a weak version of this problem is likely to be computationally intractable and so any solution to this problem will need to make essential use of the structure of the operation that the circuit is supposed to implement.
The problem we consider is, given a quantum circuit, to decide between two cases: either the circuit acts in the desired way on all input states, or the circuit misbehaves, acting in some malicious way on a large subspace of input states.
This problem is \class{QMA}-hard even when both the desired and malicious behaviour are known (i.e.\ specified by uniform families of quantum circuits).

The class \class{QMA} is the set of all problems that can be verified up to bounded error on a quantum computer.
Several problems are known to be complete for \class{QMA}: these problems can be thought of as alternate characterizations of the class, as they capture exactly the power of this computational model.
The first of these complete problems is the problem of determining the ground state energy of a local Hamiltonian.  This was first shown to be complete on $k$-local Hamiltonians~\cite{KitaevS+02} for $k \geq 5$, before the problem was shown to remain hard in the 2-local case~\cite{KempeK+06}.
The problem of determining if local descriptions of a quantum system are consistent is also known to be \class{QMA}-complete~\cite{Liu06}, though only under Turing reductions.
Other problems related to finding ground states of physical systems are also known to be complete for \class{QMA}~\cite{SchuchC+08,SchuchV09}.

There are also problems on quantum circuits that are known to be \class{QMA}-complete.
The first of these is the Non-identity check problem~\cite{JanzingW+05}, which given as input a unitary quantum circuit, the problem is to decide if there is an input on which the circuit acts non-trivially or if the circuit is close to the identity for all input states.  
The problem of determining if a circuit is close to an isometry (i.e. a reversible transformation that maps pure states to pure states) is also known to be \class{QMA}-complete~\cite{Rosgen10non-isometry}.

In this paper we generalize the hardness proofs of~\cite{JanzingW+05, Rosgen10non-isometry} to show that the \class{QMA}-hardness of the problem of testing the properties of the outputs of quantum circuits.  More specifically, we define the circuit testing problem, which has as parameters two uniformly generated families of quantum circuits $\mathcal{C}_1$ and $\mathcal{C}_2$.
The problem is do decide, given an input circuit $C$, whether $C$ acts like circuits from the family $\mathcal{C}_1$ on a large input subspace, or whether $C$ acts like circuits from $\mathcal{C}_2$ for all input states.
Using this result we reprove the \class{QMA}-hardness of non-identity check and non-isometry testing by making choices for the families $\mathcal{C}_1$ and $\mathcal{C}_2$.
We also show that some other circuit problems are hard, such asa version of finding the minimum output entropy (this is similar in spirit to the results in~\cite{BeigiS07}, though our model is incompatible), or determining when a channel has an pure (approximate) fixed point.

It is important to note that, despite the name, this problem is not related to \emph{property testing}.  
In this problem we have a significantly weaker promise---in one case the circuit only behaves in a certain way on a subspace of the input.
For an input space of dimension $d$, this subspace can be as large as $d^{1-\delta}$ for an arbitrary constant $\delta > 0$ but this subspace is still \emph{far} from the whole input space.
Essentially the problem is to detect if the circuit behaves in a certain way only when a specific input state is provided on some subset of the input qubits.
Note also that while we can use this problem to show the \class{QMA}-hardness of several circuit problems, this technique does \emph{not} show that these problems are in~\class{QMA}.

We then apply this hardness result to the problem of detecting insecure quantum encryption.  
This is the problem of deciding, given a quantum circuit that takes as input a quantum state as well as a classical key, whether this circuit is $\varepsilon$-close to a perfectly secure encryption scheme (i.e. a private quantum channel~\cite{AmbainisM+00,BoykinR03}), or whether there is a large subspace of input states that the circuit does not encrypt at all (up to error $\varepsilon$).  
To show that this problem is hard, we argue that this problem contains as a special case an instance of the circuit testing problem.  
Finally, we give a~\class{QMA} verifier for this problem to prove that it is \class{QMA}-complete.

The remainder of the paper is organized as follows: Section~\ref{scn:prelim} contains some mathematical background, a definition of the class~\class{QMA}, and a discussion of private quantum channels.  The hardness of the circuit testing problem is shown in Section~\ref{scn:circuit-testing}.  Finally, Section~\ref{scn:insecure-encryption} contains the proof that the problem of detecting insecure encryption is~\class{QMA}-complete.

\section{Preliminaries}\label{scn:prelim}

\subsection{Background}\label{scn:notation}

Throughout the paper we let $\mathcal{H,K,X,Y,\ldots}$ represent (finite-dimensional) Hilbert spaces.  
The pure quantum states are simply the unit vectors in these spaces.  
The set of density matrices on a space $\mathcal{H}$ is denoted $\density{H}$: these are the positive semidefinite operators with unit trace.
We will use the notation $\transform{H,K}$ to represent the set of channels that map states in $\density{H}$ to states in $\density{K}$.
More formally, these transformations are exactly the completely positive trace preserving linear maps from $\linear{H}$ to $\linear{K}$, where we use $\linear{H}$ to denote the set of all linear operators on $\mathcal{H}$.

To measure the distance between quantum states we will make extensive use of the trace norm, which for a linear operator $X$ can be defined as $\tnorm{X} = \tr \sqrt{X^*X}$.  
A useful alternate characterization is that $\tnorm{X}$ is the sum of the singular values of $X$, or, in the case of a normal operator, the sum of the absolute values of the eigenvalues.
One important property of the trace distance $\tnorm{\rho - \sigma}$ between two states is that it is monotone nonincreasing under the application of quantum channels.

We will also need the intuitive property that two states that are close together in the trace norm produce similar measurement outcomes.  
This can be derived from the fact that an expression involving the trace norm gives the maximum probability that two states can be distinguished~\cite{Helstrom67},
\begin{lemma}\label{lem:measurement-continuity}
  Let $X \in \linear{H}$ satisfy $0 \leq X \leq \id$.  Then
  \begin{equation*}
    \tr(X\rho)  \leq \tr(X\sigma) + \tnorm{\rho - \sigma}
  \end{equation*}
\end{lemma}

In addition to the trace norm, we will also need a distance measure on the quantum channels.
Such a measure is given by the \emph{diamond norm}, which for a linear map $\Phi: \linear{H} \to \linear{K}$ is given by $\dnorm{\Phi} = \sup_{X \in \linear{H \tprod H}} \tnorm{ (\Phi \tprod \tidentity{H})(X) } / \tnorm{X}$.  See~\cite{KitaevS+02} for an alternate definition and some further properties of this norm.
In the case that $\Phi$ is the difference of two completely positive maps, we may replace the supremum in the definition of the diamond norm with a maximization over pure states in the space $\mathcal{H \tprod H}$~\cite{RosgenW05}.
Similarly to the trace norm, the diamond norm can be used to characterize the distinguishability of two quantum channels: here the fact that the definition involves a reference system captures the fact that the optimal strategy to distinguish two channels may involve the use of entangled input states.

Since we consider computational problems on quantum channels, we must specify how they are to be given as input.
For this we use the mixed-state circuit model, first defined in~\cite{AharonovK+98}, where circuits are composed of some (universal) collection of the usual unitary gates, plus a gate that introduces ancillary qubits in the $\ket 0$ state and a gate that traces out (i.e.\ discards) qubits.  
For simplicity we will assume that all Hilbert spaces we encounter are composed of qubits, i.e.\ that the dimension is always a power of two, though this is not strictly needed.

We use this circuit model because it can (approximately) represent any quantum channel, and in the case of efficient quantum circuits this representation is of size polynomial in the number of input qubits.
Using circuits does not (significantly) restrict the applicability of our hardness results: they also apply in any model that can efficiently simulate the circuit model, such as the model of measurement based quantum computation.

\subsection{QMA}

In order to prove results about the class~\class{QMA}, we give a formal definition.
A language $L$ is in \class{QMA} if there is a quantum polynomial-time verifier $V$ such that \begin{enumerate}
  \item if $x \in L$, then there exists a witness $\rho$ such that
    $\Pr[ \text{$V$ accepts $\rho$} ] \geq 1 - \varepsilon$,
  \item if $x \not\in L$, then for any state $\rho$, 
    $\Pr[ \text{$V$ accepts $\rho$} ] \leq \varepsilon$,
\end{enumerate}
The exact value of the error parameter $\varepsilon$ is not significant: any $\varepsilon < 1/2$ that is at least an inverse polynomial in the input size suffices~\cite{KitaevS+02, MarriottW05}.

Let $L$ be an arbitrary language in \class{QMA}, and let $x$ be an arbitrary input string.  
Our goal will be to encode the \class{QMA}-hard problem of deciding if $x \in L$ into the problem of detecting an insecure encryption circuit.  
To do this it will be convenient to represent the verifier as a unitary circuit $V$, which represents the algorithm of the verifier in a \class{QMA} protocol on some input $x$.  
We may ``hard-code'' the input string $x$ into the circuit for $V$, since the circuit $V$ needs only to be efficiently generated given $x$.

The algorithm implemented by the verifier in an arbitrary \class{QMA} protocol is given in Figure~\ref{fig:qma}.  
The verifier receives a witness state $\ket\psi$, applies the unitary $V$ on the witness state and any ancillary qubits needed, and finally measures the first output qubit to decide whether or not to accept.  
Any qubits not measured are traced out.
One of the main results of this paper is a reduction from an arbitrary \class{QMA} verifier to the problem of testing the behaviour of quantum circuits.
\begin{figure}
  \begin{center}
\setlength{\unitlength}{3947sp}%
\begingroup\makeatletter\ifx\SetFigFont\undefined%
\gdef\SetFigFont#1#2#3#4#5{%
  \reset@font\fontsize{#1}{#2pt}%
  \fontfamily{#3}\fontseries{#4}\fontshape{#5}%
  \selectfont}%
\fi\endgroup%
\begin{picture}(2577,1299)(-764,-1423)
\put(-749,-436){\makebox(0,0)[lb]{\smash{{\SetFigFont{12}{14.4}{\rmdefault}{\mddefault}{\updefault}{\color[rgb]{0,0,0}$\ket \psi$}%
}}}}
\thinlines
{\color[rgb]{0,0,0}\put(751,-1036){\line( 1, 0){525}}
}%
{\color[rgb]{0,0,0}\put(751,-1111){\line( 1, 0){525}}
}%
{\color[rgb]{0,0,0}\put(751,-1186){\line( 1, 0){525}}
}%
{\color[rgb]{0,0,0}\put(751,-886){\line( 1, 0){525}}
}%
{\color[rgb]{0,0,0}\put(751,-811){\line( 1, 0){525}}
}%
{\color[rgb]{0,0,0}\put(751,-736){\line( 1, 0){525}}
}%
{\color[rgb]{0,0,0}\put(1276,-511){\oval(450,450)[tr]}
\put(1276,-511){\oval(450,450)[tl]}
}%
{\color[rgb]{0,0,0}\put(151,-1336){\framebox(600,1200){$V$}}
}%
{\color[rgb]{0,0,0}\put(-449,-286){\line( 1, 0){600}}
}%
{\color[rgb]{0,0,0}\put(-449,-361){\line( 1, 0){600}}
}%
{\color[rgb]{0,0,0}\put(-449,-436){\line( 1, 0){600}}
}%
{\color[rgb]{0,0,0}\put(-449,-511){\line( 1, 0){600}}
}%
{\color[rgb]{0,0,0}\put(-149,-961){\line( 1, 0){300}}
}%
{\color[rgb]{0,0,0}\put(-149,-1036){\line( 1, 0){300}}
}%
{\color[rgb]{0,0,0}\put(-149,-1111){\line( 1, 0){300}}
}%
{\color[rgb]{0,0,0}\put(-149,-886){\line( 1, 0){300}}
}%
{\color[rgb]{0,0,0}\put(1276,-586){\vector( 3, 4){225}}
}%
{\color[rgb]{0,0,0}\put(1051,-586){\framebox(450,450){}}
}%
{\color[rgb]{0,0,0}\put(1051,-361){\line(-1, 0){300}}
}%
{\color[rgb]{0,0,0}\put(1276,-736){\vector( 0,-1){675}}
}%
{\color[rgb]{0,0,0}\put(1501,-361){\line( 1, 0){300}}
}%
\put(-449,-1036){\makebox(0,0)[lb]{\smash{{\SetFigFont{12}{14.4}{\rmdefault}{\mddefault}{\updefault}{\color[rgb]{0,0,0}$\ket 0$}%
}}}}
{\color[rgb]{0,0,0}\put(751,-961){\line( 1, 0){525}}
}%
\end{picture}%

  \end{center}
  
  \caption{Verifier's circuit in a \class{QMA} protocol.
    The verifier accepts the witness state $\ket\psi$ if and only if
    the measurement in the computational basis results in the $\ket 1$ state.}
  \label{fig:qma}
\end{figure}
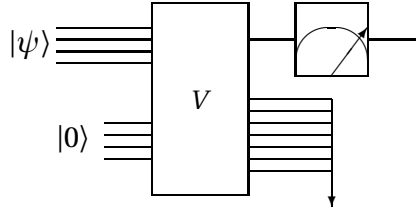

\subsection{Private Quantum Channels}\label{scn:private}

Quantum channels that are secure against eavesdroppers are those
channels for which the input state cannot be determined by the output.
These channels can also be viewed as encryption systems: the
\emph{key} is simply the environment space of the channel, which, when
combined with the output state, allows the input to be recovered.  We
restrict attention to private channels of a special form: those which
allow the input to be recovered not with the quantum state of the
environment but instead with a classical key that can be pre-shared
between two parties that wish to establish a secure quantum channel.
These channels, called, \emph{private} channels, were introduced and
studied in~\cite{AmbainisM+00,BoykinR03}.

An important example of a private quantum channel is the completely
depolarizing channel. This is the channel $\Omega$ that maps any input
to the completely mixed state.  One circuit implementation of this channel is
given in Figure~\ref{fig:depol}.
\begin{figure}
  \begin{center}
\setlength{\unitlength}{3947sp}%
\begingroup\makeatletter\ifx\SetFigFont\undefined%
\gdef\SetFigFont#1#2#3#4#5{%
  \reset@font\fontsize{#1}{#2pt}%
  \fontfamily{#3}\fontseries{#4}\fontshape{#5}%
  \selectfont}%
\fi\endgroup%
\begin{picture}(5274,1974)(289,-1498)
\put(451,-1186){\makebox(0,0)[lb]{\smash{{\SetFigFont{12}{14.4}{\rmdefault}{\mddefault}{\updefault}{\color[rgb]{0,0,0}$\ket{+}^{\otimes 6}$}%
}}}}
\thinlines
{\color[rgb]{0,0,0}\put(901,-961){\line( 1, 0){4050}}
}%
{\color[rgb]{0,0,0}\put(901,-1111){\line( 1, 0){4050}}
}%
{\color[rgb]{0,0,0}\put(901,-1186){\line( 1, 0){4050}}
}%
{\color[rgb]{0,0,0}\put(901,-1261){\line( 1, 0){4050}}
}%
{\color[rgb]{0,0,0}\put(901,-1336){\line( 1, 0){4050}}
}%
{\color[rgb]{0,0,0}\put(1426,-961){\circle*{76}}
}%
{\color[rgb]{0,0,0}\put(2026,-1036){\circle*{76}}
}%
{\color[rgb]{0,0,0}\put(2626,-1111){\circle*{76}}
}%
{\color[rgb]{0,0,0}\put(3226,-1186){\circle*{76}}
}%
{\color[rgb]{0,0,0}\put(3826,-1261){\circle*{76}}
}%
{\color[rgb]{0,0,0}\put(4426,-1336){\circle*{76}}
}%
{\color[rgb]{0,0,0}\put(1201, 14){\framebox(450,450){$X$}}
}%
{\color[rgb]{0,0,0}\put(1801,-361){\framebox(450,450){$X$}}
}%
{\color[rgb]{0,0,0}\put(2401,-736){\framebox(450,450){$X$}}
}%
{\color[rgb]{0,0,0}\put(601,239){\line( 1, 0){600}}
}%
{\color[rgb]{0,0,0}\put(601,-136){\line( 1, 0){1200}}
}%
{\color[rgb]{0,0,0}\put(601,-511){\line( 1, 0){1800}}
}%
{\color[rgb]{0,0,0}\put(3001,239){\line(-1, 0){1350}}
}%
{\color[rgb]{0,0,0}\put(3001, 14){\framebox(450,450){$Z$}}
}%
{\color[rgb]{0,0,0}\put(3601,-361){\framebox(450,450){$Z$}}
}%
{\color[rgb]{0,0,0}\put(4201,-736){\framebox(450,450){$Z$}}
}%
{\color[rgb]{0,0,0}\put(3601,-136){\line(-1, 0){1350}}
}%
{\color[rgb]{0,0,0}\put(4201,-511){\line(-1, 0){1350}}
}%
{\color[rgb]{0,0,0}\put(3451,239){\line( 1, 0){1800}}
}%
{\color[rgb]{0,0,0}\put(4051,-136){\line( 1, 0){1200}}
}%
{\color[rgb]{0,0,0}\put(4651,-511){\line( 1, 0){600}}
}%
{\color[rgb]{0,0,0}\put(1426, 14){\line( 0,-1){975}}
}%
{\color[rgb]{0,0,0}\put(3226, 14){\line( 0,-1){1200}}
}%
{\color[rgb]{0,0,0}\put(3826,-361){\line( 0,-1){900}}
}%
{\color[rgb]{0,0,0}\put(4426,-736){\line( 0,-1){600}}
}%
{\color[rgb]{0,0,0}\put(2626,-736){\line( 0,-1){375}}
}%
{\color[rgb]{0,0,0}\put(2026,-361){\line( 0,-1){675}}
}%
{\color[rgb]{0,0,0}\put(4951,-961){\vector( 0,-1){525}}
}%
{\color[rgb]{0,0,0}\put(901,-1036){\line( 1, 0){4050}}
}%
\end{picture}%

  \end{center}
  
  \caption{Example implementation of the completely depolarizing
    channel $\Omega$ on three qubits. In order to obtain a private
    channel the state the qubits in the $\ket{+}$ state are replaced by a
    classical key $k$ to obtain the channel $\Omega_k$.}
  \label{fig:depol}
\end{figure}
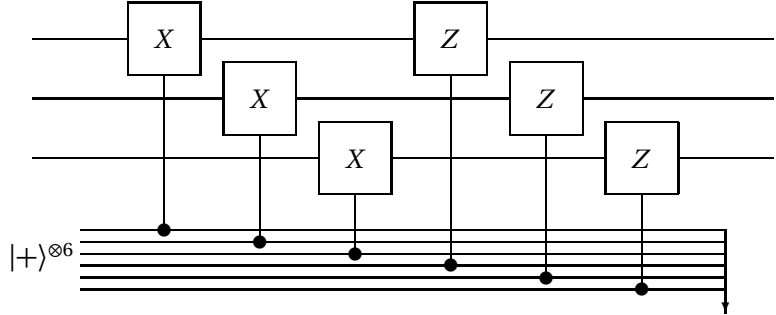

In order to use the completely depolarizing channel as a private
channel we must add a key.  This can be done to the implementation in
Figure~\ref{fig:depol} by replacing the qubits in the $\ket +$ state
with a classical string.  The result is a channel that applies a
key-specified Pauli to each of the input qubits.  We will refer to
this channel as $\Omega_k$ when a specific key is used.  Notice that
if $\Omega_k \in \transform{H}$, then $\abs{k} = 2 \log \dm{H}$, i.e.
we use two key bits for each encrypted qubit.  In the case of a
perfect encryption channel this rate of two key bits per qubit is
optimal~\cite{BraunsteinLS99, BoykinR03, AmbainisM+00}.
When the key $k$ is unknown and uniformly
distributed, the channel $\Omega_k$ is identical to $\Omega$, i.e. if
the key $k$ is uniformly distributed in $\{0, \ldots, 2^m -1\}$ we have
\begin{equation}\label{eqn:depol-key-average}
  \frac{1}{2^m} \sum_k \Omega_k = \Omega.
\end{equation}

We use the following definition of an approximately private channel
(i.e. \emph{secure encryption}).
\begin{defn}
  Let $E$ be a channel that takes two inputs: an integer $k \in \{1,
  \ldots, K\}$ and a quantum state in $\mathcal{H}$ and produces an
  output in $\mathcal{K}$, where $\dm{H} \leq \dm{K}$.  For a fixed
  value of $k$ we write $E_k(\cdot) = E(k, \cdot)$.  We call $E$ a
  \emph{$\varepsilon$-private} channel if
 \begin{enumerate}
  \item There is a decryption channel, 
    i.e. there exists a channel $D\colon \{1, \ldots, K\} \tensor
    \density{K} \to \density{H}$ such that for all $k$
    \begin{equation*}
      \dnorm{ D_k \circ E_k - \tidentity{H} } \leq \varepsilon, 
    \end{equation*}
    where the size of the circuit for $D$ is bounded by a polynomial
    in the size of the circuit for $E$.
    \label{enum-e-private-decryptable}

  \item Without the key $k$, the output of $E$ has almost no information
    about the input state, i.e.
    \begin{equation*}
      \dnorm{ \frac{1}{K} \sum_k E_k - \Omega } \leq \varepsilon
    \end{equation*}
    where $\Omega \in \transform{H,K}$ is the
    depolarizing channel that maps all inputs to
    $\identity{K}/\dm{K}$.
    \label{enum-e-private-secure}
  \end{enumerate}
\end{defn}

The use of the diamond norm in this definition is significant: we
require that both conditions hold even for part of an entangled
state.  Specifically, a channel satisfying this definition both preserves
any entanglement with the transmitted state is and remains secure even
in the case that an eavesdropper is entangled with the input.  We use
this strong definition because one of the main results of the paper is a
hardness result: distinguishing secure and insecure encryption remains
hard even when the secure encryption is promised to be secure in this
strong model.  Our hardness result remains true for the weaker model
of private channels using only the trace norm.

This definition is a strengthened
version of the model used by Ambainis and Smith~\cite{AmbainisS04}, who define security
in a similar way, but only against adversaries that are not entangled
with the input state.  Another similar model is considered by
Hayden et al.~\cite{HaydenL+04}, which also does not consider
entangled adversaries, but uses a stronger bound involving the operator
norm.  The hardness result in this paper does \emph{not} apply with
respect to this stronger bound.

Like the perfect encryption schemes found in~\cite{BoykinR03,
  AmbainisM+00}, the encryption scheme constructed by our reduction
uses $2 \log d$ key bits to encrypt a state of dimension $d$.  As
argued (implicitly) in~\cite{HaydenL+04,BarnumC+02} this is essentially optimal: any scheme
using fewer than $2 \log d (1 - \mathrm{poly}(\varepsilon))$ key bits cannot be
secure against entangled adversaries.

\section{Testing Circuits}\label{scn:circuit-testing}

The problem of testing the behaviour of a quantum circuit can be informally stated as: given a circuit $C$, decide between two cases, either the circuit acts like some known circuit $C_0$ on a large subspace of the input, or the circuit acts like some other known circuit $C_1$ on the whole input space.  
We use uniform circuit families $\mathcal{C}_0$ and $\mathcal{C}_1$ since it is important that the circuit $C$, which is provided as input, takes the same number of input and output qubits as the circuits $C_0$ and $C_1$.

\begin{problem}[\prob{Circuit Testing}]\label{prob:ct}
  Let $0 < \varepsilon < 1$, $0< \delta \leq 1$, and $\mathcal{C}_0, \mathcal{C}_1$
  be two uniform families of quantum circuits.
  The input to the problem is a circuit $C
  \in \transform{X,Y}$.  Let $C_0, C_1$ be the circuits drawn from
  $\mathcal{C}_0$ and $\mathcal{C}_1$ that take as input states on
  $\mathcal{X}$.
  The promise problem is to decide between:
  \begin{description}
  \item[Yes:] There exists a subspace $S$ of $\mathcal{X}$ with
    $\dim{S} \geq (\dm{X})^{1 - \delta}$ such that for any reference
    space $\mathcal{R}$ and any $\rho \in
    \density{\mathit{S} \tprod \mathcal{R}}$
    \begin{equation*}
      \tnorm{(C\tprod \tidentity{R})(\rho) - (C_0 \tprod
        \tidentity{R})(\rho)} \leq \varepsilon,
    \end{equation*}
    
  \item[No:] 
    $\dnorm{C - C_1} \leq \varepsilon$, i.e.\ for any reference space
    $\mathcal{R}$ and any $\rho \in \mathcal{H} \tprod \mathcal{R}$
    \begin{equation*}
      \tnorm{(C\tprod \tidentity{R})(\rho) - (C_1 \tprod
        \tidentity{R})(\rho)} \leq \varepsilon.
    \end{equation*}
  \end{description}
  
  \noindent
  When the values of $\varepsilon, \delta, \mathcal{C}_0$, and $\mathcal{C}_1$
  are significant we will refer to this problem as
  $\prob{CT}(\varepsilon, \delta, \mathcal{C}_0, \mathcal{C}_1)$.  
\end{problem}

This problem is well-defined only for families
$\mathcal{C}_0$ and $\mathcal{C}_1$ that do not violate
the promise, i.e. any circuits whose output is not too close
together.  These are the circuits $C_0$ and $C_1$ such that there does
not exist a subspace $T$ of $\mathcal{H}$ of size $\dim{T} > \dm{H}^\delta$ such
that for any input states $\rho \in \density{\mathit{T} \tprod R}$ we
have
$      \tnorm{(C_0\tprod \tidentity{R})(\rho) - (C_1 \tprod
  \tidentity{R})(\rho)} \leq 2 \varepsilon$, i.e. there does not
exist a large subspace of pure states on which $C_0$ and $C_1$ produce
output that is close together.  This condition can be
difficult to verify, but in many applications it is easy to see that
the two circuits do not agree on too many pure states.  The
application of this hardness result to detecting insecure encryption,
for instance, uses $C_0$ as the identity and $C_1$ as the completely
depolarizing channel, and these two circuits never agree on a pure
input state.  We are able to prove that this problem is
\class{QMA}-hard for any circuit families that satisfy this condition.

Notice also the special case $\delta = 1$: here
the \prob{CT} problem asks if there are \emph{any} input states on
which the circuit $C$ behaves like $C_0$ or if it behaves like $C_1$
for all input states.  In this case the problem is well-defined for
any families $\mathcal{C}_0$ and $\mathcal{C}_1$ that do not agree on
the whole space (up to error $2 \varepsilon$).

Concerning the parameters $\varepsilon$ and $\delta$, we may take
$\varepsilon = 2^{-p}$ for any polynomial $p$ using an amplification
result for \class{QMA}~\cite{MarriottW05,KitaevS+02}, and we may take
$\delta$ to be any constant satisfying $0 < \delta \leq 1$.

\subsection{Testing Circuits is QMA-hard}

To show the hardness of \prob{CT} we use a
reduction from an arbitrary problem in \class{QMA}.  This
involves embedding the verifier in a \class{QMA} protocol into an
instance of \prob{CT} with the property that the resulting circuit
runs $C_0$ if the Verifier can be made to accept and runs $C_1$ if the
Verifier cannot be made to accept.

Formalizing this notion, let $L$ be an arbitrary language in
\class{QMA} and let $x$ be an input string.  The \class{QMA}-complete
problem is to decide whether or not $x \in L$.  Since $L \in
\class{QMA}$, there exists some unitary circuit $V : \mathcal{H \tprod
A} \to \mathcal{K}$ which can be constructed efficiently from $x$ such that
if $x \in L$, there exists a pure state $\ket\psi \in \mathcal{H}$
such that measuring the first qubit of $V (\ket\psi \tprod \ket 0)$ results in $\ket 1$
with probability at least $1 - \varepsilon$, whereas
if if $x \not\in L$, then for any state $\ket\psi$ a measurement of $V
(\ket\psi \tprod \ket 0)$ results in $\ket 1$ with probability at most
$\varepsilon$.
By using standard error-reduction techniques for \class{QMA}, we may
take $\varepsilon$ to be negligible in the size of the circuit for
$V$~\cite{MarriottW05, KitaevS+02}.  Notice also that the restriction
to pure witness states $\ket\psi$ can be made without loss of
generality using a convexity argument.

Our goal is to show that \prob{CT} is hard for as many choices of
parameters as possible.  To this end, let $\delta > 0$ be constant 
and let $\mathcal{C}_0$ and $\mathcal{C}_1$ be
uniform circuit families on which the problem
$\prob{CT}(3 \sqrt{\varepsilon},\delta,\mathcal{C}_0,\mathcal{C}_1)$ is well-defined.
These are any families $\mathcal{C}_i = \{ C_{i,n} : n \geq
1 \}$, where the circuit $C_{i,n}$ takes an $n$ qubit input state, such that for any $n$ the circuits $C_{0,n}$ and $C_{1,n}$ do not produce outputs that are not too close together on some large subspace of pure input states.  
In particular, we require that
for all $n$, there does not exist a
subspace $T$ of the $n$-qubit input space $\mathcal{X}$ with $\dim{T}
> \dm{X}^\delta$ such that for any
states $\rho \in \density{\mathit{T} \tprod R}$ we
have
$ \tnorm{(C_0\tprod \tidentity{R})(\rho) - (C_1 \tprod
  \tidentity{R})(\rho)} \leq 6 \sqrt{\varepsilon}$.

The key idea to the reduction is that we construct a circuit that
takes an input state and applies the unitary $V$ to a portion of it,
makes a `copy' of the output bit with a controlled-not gate, and then
applies $V^*$.
If the result of the \class{QMA} protocol would have been the verifier
accepting (i.e. the copy of the output qubit is measured in the $\ket
1$ state), then we apply the circuit $C_0$.  On the other hand, if the
output qubit was in the $\ket 0$ state, we apply the circuit $C_1$.
This results in a circuit that applies $C_0$ if and only the input is
a state the Verifier in the \class{QMA} proof system accepts.  In
order to guarantee that the subspace of accepting states in large
enough, we add dummy input qubits that are ignored by the circuit $V$
but are acted on by either $C_0$ or $C_1$.  By adding enough of these
qubits, we can ensure that if there is at least one state $V$ accepts,
then the result is a large subspace of states that are accepted.

The full construction of the circuit produced by the reduction is
shown in Figure~\ref{fig:reduction}.  
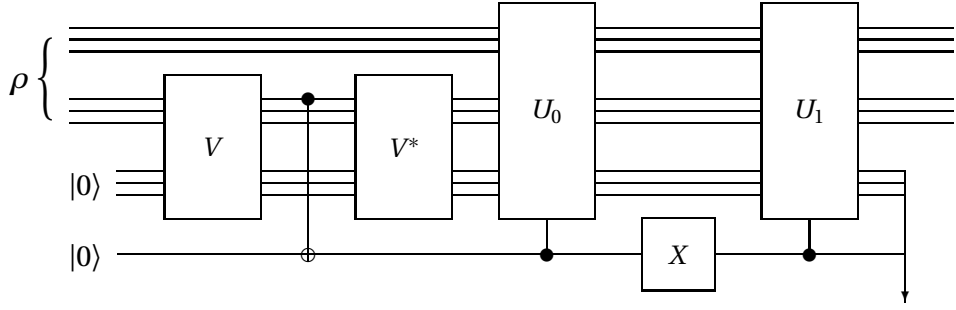
\begin{figure}
  \begin{center}
\setlength{\unitlength}{3947sp}%
\begingroup\makeatletter\ifx\SetFigFont\undefined%
\gdef\SetFigFont#1#2#3#4#5{%
  \reset@font\fontsize{#1}{#2pt}%
  \fontfamily{#3}\fontseries{#4}\fontshape{#5}%
  \selectfont}%
\fi\endgroup%
\begin{picture}(5952,1899)(-839,-1573)
\put(-449,-1336){\makebox(0,0)[lb]{\smash{{\SetFigFont{12}{14.4}{\rmdefault}{\mddefault}{\updefault}{\color[rgb]{0,0,0}$\ket 0$}%
}}}}
\thinlines
{\color[rgb]{0,0,0}\put(-449,-361){\line( 1, 0){600}}
}%
{\color[rgb]{0,0,0}\put(-449,-436){\line( 1, 0){600}}
}%
{\color[rgb]{0,0,0}\put(-449,164){\line( 1, 0){2700}}
}%
{\color[rgb]{0,0,0}\put(-449, 89){\line( 1, 0){2700}}
}%
{\color[rgb]{0,0,0}\put(-449, 14){\line( 1, 0){2700}}
}%
{\color[rgb]{0,0,0}\put(-149,-886){\line( 1, 0){300}}
}%
{\color[rgb]{0,0,0}\put(-149,-736){\line( 1, 0){300}}
}%
{\color[rgb]{0,0,0}\put(-149,-811){\line( 1, 0){300}}
}%
{\color[rgb]{0,0,0}\put(1951,-736){\line( 1, 0){300}}
}%
{\color[rgb]{0,0,0}\put(1951,-811){\line( 1, 0){300}}
}%
{\color[rgb]{0,0,0}\put(1951,-886){\line( 1, 0){300}}
}%
{\color[rgb]{0,0,0}\put(1951,-286){\line( 1, 0){300}}
}%
{\color[rgb]{0,0,0}\put(1951,-361){\line( 1, 0){300}}
}%
{\color[rgb]{0,0,0}\put(1951,-436){\line( 1, 0){300}}
}%
{\color[rgb]{0,0,0}\put(2851,-286){\line( 1, 0){1050}}
}%
{\color[rgb]{0,0,0}\put(2851,-361){\line( 1, 0){1050}}
}%
{\color[rgb]{0,0,0}\put(2851,-436){\line( 1, 0){1050}}
}%
{\color[rgb]{0,0,0}\put(2851,-736){\line( 1, 0){1050}}
}%
{\color[rgb]{0,0,0}\put(2851,-811){\line( 1, 0){1050}}
}%
{\color[rgb]{0,0,0}\put(2851,-886){\line( 1, 0){1050}}
}%
{\color[rgb]{0,0,0}\put(2851,164){\line( 1, 0){1050}}
}%
{\color[rgb]{0,0,0}\put(2851, 89){\line( 1, 0){1050}}
}%
{\color[rgb]{0,0,0}\put(2851, 14){\line( 1, 0){1050}}
}%
{\color[rgb]{0,0,0}\put(4501,-286){\line( 1, 0){600}}
}%
{\color[rgb]{0,0,0}\put(4501,-361){\line( 1, 0){600}}
}%
{\color[rgb]{0,0,0}\put(4501,-436){\line( 1, 0){600}}
}%
{\color[rgb]{0,0,0}\put(4501,164){\line( 1, 0){600}}
}%
{\color[rgb]{0,0,0}\put(4501, 89){\line( 1, 0){600}}
}%
{\color[rgb]{0,0,0}\put(4501, 14){\line( 1, 0){600}}
}%
{\color[rgb]{0,0,0}\put(4501,-736){\line( 1, 0){300}}
}%
{\color[rgb]{0,0,0}\put(4501,-811){\line( 1, 0){300}}
}%
{\color[rgb]{0,0,0}\put(4501,-886){\line( 1, 0){300}}
}%
{\color[rgb]{0,0,0}\put(2551,-1261){\circle*{76}}
}%
{\color[rgb]{0,0,0}\put(4201,-1261){\circle*{76}}
}%
{\color[rgb]{0,0,0}\put(1051,-286){\circle*{76}}
}%
{\color[rgb]{0,0,0}\put(1051,-1261){\circle{76}}
}%
{\color[rgb]{0,0,0}\put(751,-436){\line( 1, 0){600}}
}%
{\color[rgb]{0,0,0}\put(751,-361){\line( 1, 0){600}}
}%
{\color[rgb]{0,0,0}\put(751,-286){\line( 1, 0){600}}
}%
{\color[rgb]{0,0,0}\put(751,-736){\line( 1, 0){600}}
}%
{\color[rgb]{0,0,0}\put(751,-811){\line( 1, 0){600}}
}%
{\color[rgb]{0,0,0}\put(751,-886){\line( 1, 0){600}}
}%
{\color[rgb]{0,0,0}\put(151,-1036){\framebox(600,900){$V$}}
}%
{\color[rgb]{0,0,0}\put(1351,-1036){\framebox(600,900){$V^*$}}
}%
{\color[rgb]{0,0,0}\put(4801,-736){\vector( 0,-1){825}}
}%
{\color[rgb]{0,0,0}\put(2251,-1036){\framebox(600,1350){$U_0$}}
}%
{\color[rgb]{0,0,0}\put(3901,-1036){\framebox(600,1350){$U_1$}}
}%
{\color[rgb]{0,0,0}\put(2551,-1036){\line( 0,-1){225}}
}%
{\color[rgb]{0,0,0}\put(-149,-1261){\line( 1, 0){3300}}
}%
{\color[rgb]{0,0,0}\put(3151,-1486){\framebox(450,450){$X$}}
}%
{\color[rgb]{0,0,0}\put(3601,-1261){\line( 1, 0){1200}}
}%
{\color[rgb]{0,0,0}\put(4201,-1036){\line( 0,-1){225}}
}%
{\color[rgb]{0,0,0}\put(1051,-286){\line( 0,-1){1017}}
}%
\put(-449,-886){\makebox(0,0)[lb]{\smash{{\SetFigFont{12}{14.4}{\rmdefault}{\mddefault}{\updefault}{\color[rgb]{0,0,0}$\ket 0$}%
}}}}
\put(-824,-211){\makebox(0,0)[lb]{\smash{{\SetFigFont{12}{14.4}{\rmdefault}{\mddefault}{\updefault}{\color[rgb]{0,0,0}$\rho \, \Biggl\{$}%
}}}}
{\color[rgb]{0,0,0}\put(-449,-286){\line( 1, 0){600}}
}%
\end{picture}%

  \end{center}
  
  \caption{Circuit output by the reduction.  The circuit $V$ is the
    unitary circuit applied by the \class{QMA} verifier for the
    language $L$.  The circuit $U_i$ is the unitary circuit obtained
    from $C_i$ by removing the gates that introduce ancillary qubits
    and trace out qubits.}
  \label{fig:reduction}
\end{figure}
Before describing the circuit, we fix the notation that we will use.
Let $C_0$ and $C_1$ be circuits drawn from $\mathcal{C}_0$ and
$\mathcal{C}_1$ implementing transformations in $\transform{X,Y}$,
where $\mathcal{X} = \mathcal{F \tprod H}$ and $\mathcal{Y} =
\mathcal{F \tprod K}$, using the spaces $\mathcal{H,K}$ from the
\class{QMA} Verifier for $L$.  Further, we may let $\dm{F} = \ceil{
  \dm{H}^{(1-\delta)/\delta}}$, since we are free to take any
polynomial number of input qubits to $C_0$ and $C_1$.
We also assume without loss of generality that these circuits are implemented
by circuits that apply unitary circuits mapping $\mathcal{X \tprod A}
\to \mathcal{Y \tprod G}$, where the space $\mathcal{A}$ holds any
ancillary qubits needed by the circuit (initially in the $\ket 0$
state) and the space $\mathcal{G}$ represents the qubits traced out at
the end of the computation.  Any mixed-state circuit can be
efficiently transformed into a circuit of this form by moving the
introduction of ancillary qubits to the start of the circuit and
delaying any partial traces to the end of the circuit.  We may also
assume that both the circuit $V$ and the circuits $C_0$ and $C_1$ use
ancillary spaces $\mathcal{A,G}$ of the same size, by simply padding
the circuits using a smaller space with unused ancillary qubits.

Let $C$ be the circuit in Figure~\ref{fig:reduction}.  This circuit takes as input a
quantum state $\rho$ on the space $\mathcal{X} = \mathcal{F \tprod H}$.
This circuit first applies $V$ to the
portion of $\rho$ in $\mathcal{H}$ as well as any needed ancillary
qubits in the space $\mathcal{A}$.  Next, the circuit makes a
classical copy of the `output bit' of $V$, which is used as a control
for the application of the circuits $C_0$ and $C_1$.  The circuit
$V^*$ is then applied, so that the result (provided that $V$ accepts
or rejects with high probability) is a state that is close to the
input state plus a qubit that indicates whether $V$ accepts or rejects
the input state.  The circuit then applies $C_0$ if $V$ accepts and
$C_1$ if $V$ rejects.  These circuits use the same ancillary space
$\mathcal{A}$ as the circuits $V$ and $V^*$, but as long as the
Verifier $V$ either accepts of rejects the input state with high
probability, these ancillary qubits will be returned to the $\ket 0$
state, up to trace distance $2 \sqrt{\varepsilon}$.

Before proving the correctness of the reduction, it will be convenient
to write down some of the states produced by running the constructed
circuit $C$.
Let $\rho$ be an arbitrary input state in $\density{H \tprod F}$ and let
$\ket\psi \in \mathcal{H \tprod F \tprod R}$ be a purification of $\rho$.
The order of the spaces $\mathcal{H}$ and $\mathcal{F}$ has been
changed for notational convenience.
After applying the unitary $V$ to the portion of $\ket\psi$ in
$\mathcal{H}$, the state can be written as
\begin{equation*}
  \ket\phi = (V \tprod \identity{F} \tprod \identity{R}) (\ket\psi \tprod \ket 0),
\end{equation*}
where the $\ket 0$ qubits are in the space $\mathcal{A}$.  Then, there
exist states $\ket{\phi_0}, \ket{\phi_1}$ on all but the first qubit
of $\mathcal{K \tprod F \tprod R}$ such that
\begin{equation*}
  \ket\phi = \sqrt{1-p} \ket 0 \tprod \ket{\phi_0} + \sqrt{p} \ket 1 \tprod \ket{\phi_1}
\end{equation*}
where $0 \leq p \leq 1$ is exactly the probability that the Verifier accepts in the
original protocol on input $\ptr{F} \rho$.  Applying the
controlled-not gate results in
\begin{equation*}
  \ket{\phi'} = \sqrt{1-p} \ket{00} \tprod \ket{\phi_0} + \sqrt{p} \ket{11} \tprod \ket{\phi_1}.
\end{equation*}
We then bound the trace distance of $\ket{\phi'}$ to
$\ket 0 \ket\psi$ and $\ket 1 \ket\psi$.  In the case of $\ket
0 \ket \psi$ we have
\begin{equation}\label{eqn:reduction-close-0}
  \tnorm{ \ketbra{\phi'} - \ketbra{0} \tprod \ketbra{\phi}}
  = 2 \sqrt{ 1 - \abs{ \braket{\phi'}{0 \phi} }^2 }
  = 2 \sqrt{ 1 - (1-p)^2 }
  < 3 \sqrt{p},
\end{equation}
and in the similar case of $\ket 1 \ket \psi$ we have
\begin{equation}\label{eqn:reduction-close-1}
  \tnorm{ \ketbra{\phi'} - \ketbra{1} \tprod \ketbra{\phi}}
  = 2 \sqrt{ 1 - \abs{ \braket{\phi'}{1 \phi} }^2 }
  = 2 \sqrt{ 1 - p^2 }
  < 3 \sqrt{1-p}.
\end{equation}
These two equations show that, when $p$ is close to $0$ or $1$, the
fact that we make a classical copy of the output qubit does not have a
large effect on the state of the system.  (This fact can also be argued
from the Gentle Measurement Lemma~\cite{Winter99}.)
The remainder of the circuit then applies $V^*$ and, depending on the
value of the control qubit, one of $C_0$ and $C_1$.  We consider two
cases, which are argued in two separate propositions.

\begin{proposition}\label{prop:reduction-yes-bound}
  If $x \in L$, then there exists a subspace $S$ of $\mathcal{X}$ with
  $\dim{S} \geq \dm{X}^{1-\delta}$ such that for any reference system
  $\mathcal{R}$ and any $\rho \in S \tprod \mathcal{R}$
  \begin{equation}
    \tnorm{ (C \tprod \tidentity{R})(\ketbra\psi) - (C_0 \tprod
      \tidentity{R})(\ketbra\psi }
    \leq 3 \sqrt{ \varepsilon }.
  \end{equation}
\end{proposition}
\begin{proof}
  If $x\in L$, then there is some input state $\ket\psi$ on which the
  Verifier accepts with probability $p \geq 1 - \varepsilon$.
  Applying the remainder of the circuit, up to the partial trace, to
  the state $\ket 1 \ket\phi$ results in the state $\ket 1 \tprod (U_1
  \tprod \identity{R})( \ket \psi \tprod \ket 0)$.  Tracing out the
  space $\mathcal{G}$ as well as the copy of the output qubit, results
  in exactly the state $\ptr{G} (U_1 \tprod
  \identity{R})(\ketbra{\psi} \tprod \ketbra{0}) (U_1^* \tprod
  \identity{R}) = (C_1 \tprod \tidentity{R})(\ketbra\psi)$.  This is
  not quite equal to the output of the constructed circuit $C$,
  however, as in this evaluation we have replaced the state
  $\ket{\phi'}$ with the state $\ket 1 \ket \phi$.  However, using the
  monotonicity of the trace norm under quantum operations, the
  remainder of the circuit cannot increase the norm of the two states,
  and so applying Equation~\eqref{eqn:reduction-close-1}, we have
  \begin{equation}\label{eqn:reduction-yes-bound}
    \tnorm{ (C \tprod \tidentity{R})(\ketbra\psi) - (C_0 \tprod
      \tidentity{R})(\ketbra\psi }
    \leq 3 \sqrt{1 - p} \leq 3 \sqrt{ \varepsilon }.
  \end{equation}

  It remains to show that this occurs on a large subspace of
  $\mathcal{X} = \mathcal{H} \tprod \mathcal{F}$.  Since we have
  assumed the Verifier $V$ accepts with high probability on the state
  $\ket\psi$, this implies that there is some state $\ket \gamma \in
  \mathcal{H}$ for which $V$ also accepts with probability at least $1
  - \varepsilon$, as $V$ ignores the qubits in $\mathcal{F}$.  
  Then, since $\ket\psi$ was arbitrary,
  Equation~\eqref{eqn:reduction-yes-bound} also applies to $\ket\gamma
  \tprod \ket\xi \in \mathcal{H \tprod F}$ for \emph{any} state
  $\ket\xi \in \mathcal{F}$.  The subspace $S$ of states whose reduced
  state on $\mathcal{H}$ is equal to $\ket\gamma$ has dimension
  $\dm{F}$.  Then, since $\dm{F} = \ceil{
    \dm{H}^{(1-\delta)/\delta}}$, we have
  \begin{equation*}
    \dm{X} = \dm{H \tprod F} = \dm{H} \dm{F}
    \leq \dm{F}^{\delta/(1-\delta)} \dm{F}
    = \dm{F}^{1/(1-\delta)},
  \end{equation*}
  which implies that $\dm{F} \geq \dm{X}^{1-\delta}$, as required.
  Thus, when $x \in L$ the Verifier $V$ can be made to accept, and
  so the result is a yes instance of \prob{CT}.
\end{proof}

The remaining case is when $x \not\in L$, i.e.\ the Verifier $V$
rejects every state with high probability.  This proof of this case is extremely
similar to the previous one.
\begin{proposition}\label{prop:reduction-no-bound}
  If $x \not\in L$, then for any reference system
  $\mathcal{R}$ and any $\rho \in \mathcal{X} \tprod \mathcal{R}$,
   $\dnorm{ C - C_1}
    \leq 3 \sqrt{ \varepsilon }.$
\end{proposition}
\begin{proof}
  This proof is similar to the proof of
  Proposition~\ref{prop:reduction-yes-bound}.  
  If $x\not\in L$, then $V$ accepts any
  state $\ket\psi$ with probability $p \leq \varepsilon$.  If we
  consider applying $V^*$ and the remainder of the circuit to the state
  $\ket 0 \ket\phi$, the result is $(C_1 \tprod
  \tidentity{R})(\ketbra\psi)$, similarly to the previous case.  Once
  again, we do not run this part of the circuit on this state, but the
  state $\ket{\phi'}$ which is very close to it.  Once again we can
  apply the monotonicity of the trace norm under quantum operations and Equation~\eqref{eqn:reduction-close-0} to show that 
  \begin{equation*}
    \tnorm{ (C \tprod \tidentity{R})(\ketbra\psi) - (C_1 \tprod
      \tidentity{R})(\ketbra\psi }
    \leq 3 \sqrt{p} \leq 3 \sqrt{ \varepsilon }.
  \end{equation*}
  Since this equation applies for all reference systems $\mathcal{R}$
  and all states $\ket\psi$, this proves that if $x \not\in L$, then we
  have $\dnorm{C - C_1} \leq 3 \sqrt{\varepsilon}$.
\end{proof}

Taken together, these two proposition prove the hardness of the
\prob{CT} problem.  Note once again that in order for the \prob{CT}
problem to be well defined (i.e. the set of `yes' instances does not
intersect the set of `no' instances) we require that circuits from the
two families are not too close together for any large subspaces of pure
input states.  See the discussion following Problem~\ref{prob:ct} for
a technical condition that is equivalent to this requirement.
\begin{theorem}\label{thm:cpt-hardness}
  \prob{CT}($\varepsilon,\delta,\mathcal{C}_0,\mathcal{C}_1$) is \class{QMA}-hard for any $0 < \varepsilon < 1$ such that $\varepsilon \geq 2^{-p}$ for some polynomial $p$, any constant $0 < \delta \leq 1$, and any uniform circuit families $\mathcal{C}_0$, $\mathcal{C}_1$ for which the problem is well-defined.
\end{theorem}
\begin{proof}
  The correctness of the reduction is argued in In Propositions~\ref{prop:reduction-yes-bound}
  and~\ref{prop:reduction-no-bound}.
  It remains only to verify that the reduction can be
  performed efficiently.  To see that the reduction can be performed in
  time polynomial in the size of the input $x$ (which is at most
  polynomially smaller than the size of the circuit
  $V$: the only part of the reduction that can cause a problem the size of
  the space $\mathcal{F}$, since we have taken $\dm{F} =
  \ceil{\dm{H}^{(1-\delta)/\delta}}$.  This implies that the space
  $\mathcal{F}$ requires a factor of $(1-\delta) / \delta$ more qubits
  than the space $\mathcal{H}$, which is linear in the input dimension so long as
  $\delta$ is a constant.  This implies
  that the reduction can be performed in (classical deterministic)
  polynomial time.
\end{proof}

\subsection{Applications}

In this section we apply Theorem~\ref{thm:cpt-hardness} to reprove the hardness of some of the circuit problems that are known to be hard for \class{QMA} as well as to show the \class{QMA}-hardness of some new circuit problems.

The first problem we consider is a slightly generalized version of the
problem~\prob{Non-identity Check} studied by Janzing, Wocjan, and
Beth~\cite{JanzingW+05}, who show that it is \class{QMA}-complete.
Our version of the problem differs in that
we allow the input circuit to be a mixed-state circuit.  We do still
require, however, that if the circuit does not act like the identity
everywhere, then it acts like some efficient unitary circuit $U$ on some
input state for which $U$ is far from the identity.  This requirement
is not needed to prove that this problem is hard, but it is hard to
see how to put the problem into~\class{QMA} without it.
\begin{problem}[\prob{Mixed Non-identity Check}~\cite{JanzingW+05}]\label{prob:nonidentity}
  Let $0 < \varepsilon < 1$.  On input $C$, a circuit in $\in \transform{X,X}$,
  the promise problem is to decide between:
  \begin{description}
  \item[Yes:]  $\dnorm{C - \tidentity{}} \geq 2 - \varepsilon$ and there
    exists an efficient unitary $U$ such that on some pure state
    $\ket\psi \in \mathcal{X}$ we have $\tnorm{ C(\ketbra{\psi}) - U
      \ketbra{\psi} U^*} \leq \varepsilon$ and $\tnorm{ U
      \ketbra{\psi} U^* - \ketbra{\psi}} \geq 2 - \varepsilon$.

  \item[No:] 
    $\dnorm{C - \tidentity{}} \leq \varepsilon$.
  \end{description}
\end{problem}
\noindent The \class{QMA}-hardness of this problem follows from
Theorem~\ref{thm:cpt-hardness} and the fact that
\prob{CT}$(\varepsilon,1,\mathcal{U},\id)$ is a special case of the problem, where
$\mathcal{U}$ is any uniform family of quantum circuits that are not
close to the identity (one such example is the family of circuits that apply Pauli $X$ to
the first input qubit).

The next problem we consider is the problem of detecting whether a (mixed-state) circuit is close to an isometry, which was shown to be \class{QMA}-complete in~\cite{Rosgen10non-isometry}.  This can be formalized as the problem of detecting if there is a pure input state one which the output state is highly mixed.
\begin{problem}[\prob{Non-isometry}~\cite{Rosgen10non-isometry}]\label{prob:noniso}
  Let $0 < \varepsilon < 1/2$.  On input a circuit $C \in
  \transform{X,Y}$ the promise
  problem is to decide between:
  \begin{description}
    \item[Yes:] There exists $\ket\psi \in \mathcal{X}$ such that $\opnorm{(\Phi
        \tprod \tidentity{X})(\ketbra{\psi})} \leq \varepsilon$,
    \item[No:] For all $\ket\psi \in \mathcal{X}$, $\opnorm{(\Phi
        \tprod \tidentity{X})(\ketbra{\psi})} \geq 1 - \varepsilon$.
  \end{description}
\end{problem}
\noindent The \class{QMA}-hardness of this problem follows from Theorem~\ref{thm:cpt-hardness}, since \prob{CT}$(\varepsilon,1,\Omega,\id)$ is a special case.

 We can also apply Theorem~\ref{thm:cpt-hardness} to show the hardness of the problem of determining if a channel has a pure fixed point.  This problem can be stated as follows. 
\begin{problem}[\prob{Pure Fixed Point}]\label{prob:pfp}
  Let $0 < \varepsilon < 1$.  On input a circuit $C \in
  \transform{X,X}$ the promise
  problem is to decide between:
  \begin{description}
  \item[Yes:] There exists $\ket\psi \in \mathcal{X}$ such that $\tnorm{C(\ketbra{\psi}) - \ketbra{\psi}} \leq \varepsilon$
  \item[No:] For any $\ket\psi \in \mathcal{X}$, $\tnorm{C(\ketbra{\psi}) - \ketbra{\psi}} \geq 2 - \varepsilon$
  \end{description}
\end{problem}    
\noindent The \class{QMA}-hardness of this problem follows from the fact that \prob{CT}$(\varepsilon,1,\id,\Omega)$ is a special case. 

A related problem is determining if the minimum output entropy of a quantum channel is small.  Related results can be found in~\cite{BeigiS07}, though the model used there seems to be incompatible with the model used in the present paper.
In order to define this problem, let $S_{\min}(C) = \min_{\rho} S(C(\rho))$ be the minimum output entropy of the channel $C$ (where $S$ is the von Neumann entropy).
\begin{problem}[\prob{Minimum Output Entropy}]\label{prob:moe}
  Let $0 < \varepsilon < 1/2$.  On input a circuit $C \in
  \transform{X,X}$ the promise
  problem is to decide between:
  \begin{description}
  \item[Yes:] $S_{\min}(C) \leq \varepsilon \log \dm{X}$
  \item[No:] $S_{\min}(C) \geq (1 - \varepsilon) \log \dm{X}$
  \end{description}
\end{problem}
\noindent As in the previous case, the \prob{QMA}-hardness of this problem follows from Theorem~\ref{thm:cpt-hardness} and the fact that \prob{CT}$(\varepsilon/2,1,\id,\Omega)$ is a special case.  The $\log \dm{X}$ terms in the statement of the problem are due to the use of Fannes Inequality~\cite{Fannes73} to transform trace distance bounds to entropy bounds.

\section{Detecting Insecure Encryption}\label{scn:insecure-encryption}

In this section we consider the problem of detecting when a two-party symmetric key quantum encryption system is insecure.
We first use Theorem~\ref{thm:cpt-hardness} to show that this problem is hard, and then give a \class{QMA}-verifier to show that it is \class{QMA}-complete.
The problem can be defined as follows.

\begin{problem}[\prob{Detecting Insecure Encryption}]\label{prob:insecure}
 For $0 < \varepsilon < 1$ and $0 < \delta \leq 1$ an instance of the problem consists of
 a quantum circuit $E$ that takes as input a quantum state as well as
 a $m$ classical bits, such that for each $k \in \{0,1\}^m$ the
 circuit implements a quantum channel $E_k \in \transform{H,K}$ with
 $\dm{K} \geq \dm{H}$.  The
 promise problem is to decide between:
  \begin{description}
    \item[Yes:] There exists a subspace $S$ of $\mathcal{H}$ with
      $\dim{S} \geq \dm{H}^{1-\delta}$ such that for any reference
      space $\mathcal{R}$, any $\rho \in
      \density{\mathit{S} \tprod R}$, and
      any key $k$,
      $\tnorm{(E_k \tprod \tidentity{R})(\rho) - \rho} \leq \varepsilon.$

    \item[No:] $E$ is an $\varepsilon$-private channel, i.e. 
      $\smalldnorm{ \Omega - \frac{1}{2^m} \sum_{k \in \{0,1\}^m} E_k } \leq \varepsilon, $
      where $\Omega$ is the completely depolarizing channel in
      $\transform{H,K}$, and there exists an polynomial-size quantum
      circuit $D$ such that for all $k$ we have $
      \dnorm{D_k \circ E_k - \tidentity{H}} \leq \varepsilon$.
  \end{description}
  When the values of $\varepsilon$ and $\delta$ are significant, we will refer to this
  problem as $\prob{DI}_{\varepsilon,\delta}$.  
\end{problem}

Informally, this is the problem of distinguishing two cases: either
the channel fails to encrypt a large subspace of the input qubits (for
any key), or the channel is very close to a perfect encryption
channel.  

\begin{theorem}\label{thm:hardness}
  $\prob{DI}_{\varepsilon,\delta}$ is \class{QMA}-hard for all $0 <
  \varepsilon < 1/2$ and all $0 < \delta \leq 1$.  
\end{theorem}
\begin{proof}
  Let $\mathcal{E}_k = \{ \Omega_{k,n} \}$ where $\Omega_{k,n}$ is the
  $n$-qubit channel that applies the $k$th Pauli operator to the input
  qubits.  As in Equation~\eqref{eqn:depol-key-average} averaging over
  all over all keys $k$ results in the completely depolarizing channel
  on $n$ qubits.  Then, Theorem~\ref{thm:cpt-hardness} implies that  
  \prob{CT}$(\varepsilon, \delta, \id_k, \mathcal{E}_k)$ is hard for \class{QMA}, 
  where $\id_k$ is the channel that
  discards the key $k$ and does nothing to the quantum input.

  The problem \prob{CT}$(\varepsilon, \delta, \id_k, \mathcal{E}_k)$
  involves a slight redefinition of the problem \prob{CT} to include
  both a quantum input, as well as a classical input $k$.  This can be
  done without difficulty by including the classical input as part of
  the quantum input (to circuits in the families $\id_k$ and
  $\mathcal{E}_k$) that is immediately measured in the computational
  basis (and in the case of $\id_k$, discarded).  The problem
  \prob{CT}$(\varepsilon, \delta, \id_k, \mathcal{E}_k)$ remains hard
  after this modification.

  The \class{QMA}-hardness of $\prob{DI}_{\varepsilon,\delta}$ then
  follows immediately from the fact that the problem of detecting
  insecure encryption is simply \prob{CT}$(\varepsilon, \delta, \id_k,
  \mathcal{E}_k)$ with a weakened promise.  
  Since the sets of `yes' instances of the two
  problems are identical, we need only verify the `no' instances.  Let
  the circuit $C \in \transform{H, K}$ be a `no' instance of
  \prob{CT}$(\varepsilon, \delta, \id_k, \mathcal{E}_k)$ and let
  $C_k(\cdot) = C( \ketbra{k} \tprod \cdot)$ be the circuit defined by
  hardcoding the input in the `key' portion of the input space.  Then,
  for any input $\rho$ and any key $k$, we have $\dnorm{C_k -
    \Omega_k} \leq \varepsilon$, since this follows for the versions
  of these circuits without a hardcoded key (which is just a
  restriction of the input space).  From this equation, the triangle
  inequality implies that
  \begin{equation*}
    \dnorm{ \Omega - \frac{1}{2^m} \sum_k C_k }
    \leq \frac{1}{2^m} \sum_k \dnorm{ \Omega_k - C_k }
    \leq \varepsilon,
  \end{equation*}
  which is the property required by `no' instances of \prob{DI}.  To
  see further that the output of $C_k$ can be decrypted with
  knowledge of $k$, observe that $\Omega^{-1}_k \circ \Omega_k =
  \tidentity{}$, and so it follows that
  \begin{equation*}
    \dnorm{ \Omega^{-1}_k \circ C_k - \tidentity{}}
    \leq \dnorm{ \Omega^{-1}_k \circ C_k - \Omega^{-1}_k \circ
      \Omega_k } + \dnorm{ \Omega^{-1}_k \circ \Omega_k -
      \tidentity{}} \\
    \leq \dnorm{C_k - \Omega_k } \leq \varepsilon,    
  \end{equation*}
  which implies that instances of \prob{CT}$(\varepsilon, \delta,
  \id_k, \mathcal{E}_k)$ are equivalent to instances of
  $\prob{DI}_{\varepsilon,\delta}$, as required.
\end{proof}

\subsection{QMA Protocol}

To test the security of an encryption system in \class{QMA} the
Verifier will need a tool to compare two quantum states.  Such a tool
is provided by the swap test, introduced in~\cite{BuhrmanC+01}, though
here we essentially use it to test the purity of quantum states as is
done in~\cite{EkertA+02}.

The swap test is an efficient procedure that makes the projective
measurement onto the symmetric and antisymmetric subspaces of a
bipartite space.  Let $W$ be the swap operation on $\mathcal{H \tprod
  H}$, i.e. $W (\ket{\psi} \tprod \ket{\phi}) = \ket{\phi} \tprod
\ket{\psi}$ for all $\ket\psi, \ket\phi \in \mathcal{H}$.  The swap
test performs the two-outcome projective measurement given by the projection
onto the symmetric subspace, given by $(\identity{H \tprod H} + W)/2$,
and the projection onto the antisymmetric subspace, given by
$(\identity{H \tprod H} - W)/2$.  

Given two pure states $\ket\psi,
\ket\phi$, the swap test returns the symmetric outcome with
probability $(1 + \abs{ \braket{\psi}{\phi} }^2)/2$.
When applied to mixed states $\rho,\sigma$, the swap test can also be
used to estimate the overlap, as the result is symmetric with
probability $(1 + \tr( \rho \sigma ))/2$, as observed
in~\cite{EkertA+02}.  Notice that this implies that the swap test can
be used to estimate the purity of a state, given two copies.

The idea behind the protocol is that if the encryption system
specified by $E$ is insecure then, regardless of the key chosen, it
acts trivially on some subspace of the input states.  In this case a
proof can consist simply of two copies of some pure state in this
subspace.  The Verifier runs $E$ on both of these states in parallel
and tests that they have not been changed by performing the swap
test.  In the case that the circuit
is insecure, this proof state will cause the Verifier to obtain the
symmetric outcome of the swap test with probability approaching 1.
Note that this protocol does \emph{not} check that the input state
is unchanged, only that the output states of the two applications of
$E$ are (close to) the same pure state.

If $E$ represents a secure encryption system, then without knowledge
of the key, the output of $E$ is close to the completely mixed
state, regardless of the input state.  In this case the Verifier
performs the swap test on two highly mixed states and the result is
antisymmetric with probability close to 1/2.

This protocol can be formalized as follows.  A circuit implementation
can be found in Figure~\ref{fig:protocol}.
\begin{proto}\label{proto:di-qma}
  On input a circuit $E \colon \{1, \ldots, K \} \tensor \density{H}
  \to \density{K}$, an instance of $\prob{DI}_{\varepsilon,\delta}$, as
  well as a quantum proof $\ket\phi$ in $\density{(H \tensor R)^{\mathrm{\tprod
  2}}}$ (where $\dm{R} = \dm{H}$), the
  Verifier performs the following protocol.
  \begin{enumerate}
    \item The Verifier generates random keys $k_1, k_2 \in \{1, \ldots,
      K\}$.
    \item The Verifier applies $(E_{k_1} \tprod \tidentity{R}) \tprod
      (E_{k_2} \tprod \tidentity{R})$ to the state
      $\ket\phi$.
    \item The Verifier applies the swap test to the resulting state,
      accepting if the outcome is symmetric.
  \end{enumerate}
\end{proto}
\begin{figure}
  \begin{center}
   \setlength{\unitlength}{3947sp}%
\begingroup\makeatletter\ifx\SetFigFont\undefined%
\gdef\SetFigFont#1#2#3#4#5{%
  \reset@font\fontsize{#1}{#2pt}%
  \fontfamily{#3}\fontseries{#4}\fontshape{#5}%
  \selectfont}%
\fi\endgroup%
\begin{picture}(3027,2949)(1186,-2698)
\put(1201,-961){\makebox(0,0)[lb]{\smash{{\SetFigFont{12}{14.4}{\rmdefault}{\mddefault}{\updefault}{\color[rgb]{0,0,0}$\ket 0$}%
}}}}
\thinlines
{\color[rgb]{0,0,0}\put(2251,-2236){\line( 1, 0){300}}
}%
{\color[rgb]{0,0,0}\put(2251,-2386){\line( 1, 0){300}}
}%
{\color[rgb]{0,0,0}\put(1201,-1786){\line( 1, 0){1350}}
}%
{\color[rgb]{0,0,0}\put(1201,-1711){\line( 1, 0){1350}}
}%
{\color[rgb]{0,0,0}\put(1201,-1861){\line( 1, 0){1350}}
}%
{\color[rgb]{0,0,0}\put(1501,-2311){\line( 1, 0){300}}
}%
{\color[rgb]{0,0,0}\put(1501,-2236){\line( 1, 0){300}}
}%
{\color[rgb]{0,0,0}\put(1501,-2386){\line( 1, 0){300}}
}%
{\color[rgb]{0,0,0}\put(1201,-1411){\line( 1, 0){2100}}
}%
{\color[rgb]{0,0,0}\put(1201,-1336){\line( 1, 0){2100}}
}%
{\color[rgb]{0,0,0}\put(1201,-1486){\line( 1, 0){2100}}
}%
{\color[rgb]{0,0,0}\put(1201,-361){\line( 1, 0){1350}}
}%
{\color[rgb]{0,0,0}\put(1201,-286){\line( 1, 0){1350}}
}%
{\color[rgb]{0,0,0}\put(1201,-436){\line( 1, 0){1350}}
}%
{\color[rgb]{0,0,0}\put(1501,-886){\line( 1, 0){300}}
}%
{\color[rgb]{0,0,0}\put(1501,-811){\line( 1, 0){300}}
}%
{\color[rgb]{0,0,0}\put(1501,-961){\line( 1, 0){300}}
}%
{\color[rgb]{0,0,0}\put(2251,-886){\line( 1, 0){300}}
}%
{\color[rgb]{0,0,0}\put(2251,-811){\line( 1, 0){300}}
}%
{\color[rgb]{0,0,0}\put(2251,-961){\line( 1, 0){300}}
}%
{\color[rgb]{0,0,0}\put(1201, 14){\line( 1, 0){2100}}
}%
{\color[rgb]{0,0,0}\put(1201, 89){\line( 1, 0){2100}}
}%
{\color[rgb]{0,0,0}\put(1201,-61){\line( 1, 0){2100}}
}%
{\color[rgb]{0,0,0}\put(3001,-1786){\line( 1, 0){300}}
}%
{\color[rgb]{0,0,0}\put(3001,-1861){\line( 1, 0){300}}
}%
{\color[rgb]{0,0,0}\put(3001,-1711){\line( 1, 0){300}}
}%
{\color[rgb]{0,0,0}\put(3001,-1636){\line( 1, 0){300}}
}%
{\color[rgb]{0,0,0}\put(3001,-361){\line( 1, 0){300}}
}%
{\color[rgb]{0,0,0}\put(3001,-436){\line( 1, 0){300}}
}%
{\color[rgb]{0,0,0}\put(3001,-286){\line( 1, 0){300}}
}%
{\color[rgb]{0,0,0}\put(3001,-211){\line( 1, 0){300}}
}%
{\color[rgb]{0,0,0}\put(1801,-2536){\framebox(450,450){$\Omega$}}
}%
{\color[rgb]{0,0,0}\put(2551,-2536){\framebox(450,975){$E$}}
}%
{\color[rgb]{0,0,0}\put(1801,-1111){\framebox(450,450){$\Omega$}}
}%
{\color[rgb]{0,0,0}\put(2551,-1111){\framebox(450,975){$E$}}
}%
{\color[rgb]{0,0,0}\put(3301,-2686){\framebox(600,2925){\parbox{0.6in}{\centering{swap\\ test}}}}
}%
{\color[rgb]{0,0,0}\put(3901,-1186){\line( 1, 0){300}}
}%
\put(1201,-2386){\makebox(0,0)[lb]{\smash{{\SetFigFont{12}{14.4}{\rmdefault}{\mddefault}{\updefault}{\color[rgb]{0,0,0}$\ket 0$}%
}}}}
{\color[rgb]{0,0,0}\put(2251,-2311){\line( 1, 0){300}}
}%
\end{picture}%

  \end{center}
  
  \caption{The Verifier's circuit in the \class{QMA} protocol.}
  \label{fig:protocol}
\end{figure}
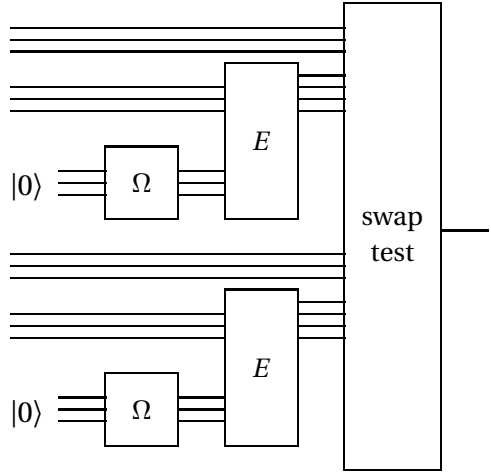

The reference space $\mathcal{R}$ appears in this protocol, but
Problem~\ref{prob:insecure} places no upper bound on the size of this space, and the
value of the norm being verified may increase with the size of the
space $\mathcal{R}$.  Fortunately, this process stabilizes when
$\dm{R} = \dm{H}$, and so we may assume that this space is of this
size, which at most doubles the number of input qubits to the protocol. 

A straightforward argument based on the continuity of measurement
probabilities (here given as Lemma~\ref{lem:measurement-continuity})
can be used to show that this protocol is correct.
\begin{proposition}
  For $0 < \varepsilon < 1/8$, Protocol~\ref{proto:di-qma} is a
  \class{QMA} protocol for $\prob{DI}_{\varepsilon,\delta}$.
\end{proposition}
\begin{proof}
  If $E$ is a `yes' instance of $\prob{DI}_{\varepsilon,\delta}$, then there
  exists a state $\ket\psi \in \mathcal{H \tprod R}$ such that for any key $k
  \in \{1, \ldots, K\}$ we have
  $\tnorm{\hat{E}_k(\ketbra\psi) - \ketbra\psi} \leq \varepsilon$,
  where throughout this proof we use the shorthand notation $\hat{E}_k
  = E_k \tprod \tidentity{R}$.  
  Let
  the input state be $\ket\phi = \ket\psi \tprod \ket \psi$.
  Fixing notation further, let $\hat{E}_k(\ketbra\psi) = \sigma_k$.  Applying
  $\hat{E}_{k_1} \tprod \hat{E}_{k_2}$ to $\ket\psi \tprod \ket\psi$ results in a
  state $\sigma_{k_1} \tprod \sigma_{k_2}$ that satisfies
  \begin{equation}
    \tnorm{\sigma_{k_1} \tprod \sigma_{k_2} - \ketbra{ \psi } \tprod
      \ketbra{ \psi } } \leq 2 \varepsilon,
  \end{equation}
  which follows from the triangle inequality.
  Then, since the state $\ketbra{ \psi } \tprod \ketbra{ \psi }$ is
  symmetric and we can view the swap test can be viewed as a projective
  measurement, Lemma~\ref{lem:measurement-continuity} shows that the
  swap test returns the symmetric outcome on $\sigma_{k_1} \tprod
  \sigma_{k_2}$ with probability at least $1 - 2 \varepsilon$.  This
  implies that when the circuit $E$ is not secure 
  the Verifier accepts with high probability.
  
  It remains to show that when the circuit $E$ is a `no' instance of
  $\prob{DI}_{\varepsilon,\delta}$ the Verifier does not accept any proof state
  with high probability.  In this case we know that $\smalldnorm{ \sum_{k=1}^K
    E_k - \Omega} / K \leq \varepsilon$.  Once more, a straightforward
  argument using the triangle inequality can be used to argue that the tensor product of two copies satisfies the equation
  $ \smalldnorm{ \sum_{k,j=1}^K E_k \tprod E_j - \Omega
  \tprod \Omega} / K^2 \leq 2 \varepsilon$.
  This implies that regardless of the proof state $\ket\psi$ the input
  to the swap test is within trace distance $2 \varepsilon$ of the
  completely mixed state.  On such a state,
  Lemma~\ref{lem:measurement-continuity} implies that 
  the swap test returns the
  symmetric outcome with probability at most
  \begin{equation*}
    \frac{1}{2} - \frac{1}{2} \tr \left[ \left(
        \frac{\identity{K}}{\dm{K}} \right)^2 \right] + 2 \varepsilon
    = \frac{1}{2} - \frac{1}{2 \dm{K}} + 2 \varepsilon,
  \end{equation*}
  and so the probability the Verifier accepts is bounded
  above by $1/2 + 2 \varepsilon$.  Thus, when $\varepsilon < 1/8$,
  there is a constant gap between the acceptance probabilities in the
  two cases, and so $\prob{DI}_{\varepsilon,\delta} \in \class{QMA}$.
\end{proof}

Combining the previous Proposition with Theorem~\ref{thm:hardness} we
obtain the main result.
\begin{theorem}
  For $0 < \varepsilon < 1/8$ and $0 < \delta \leq 1$, the problem $\prob{DI}_{\varepsilon,\delta}$ is \class{QMA}-complete.
\end{theorem}

\section{Discussion}

We have shown the \class{QMA}-hardness of a general version of the problem of testing the behaviour of a quantum circuit.
This result generalizes the proofs of hardness for many of the known circuit problems that are \class{QMA}-hard~\cite{JanzingW+05,Rosgen10non-isometry}, as well as allows for simple proofs of hardness for new circuit problems.
As an application of this result we have shown that the problem of detecting insecure encryption is complete for \class{QMA} by in addition finding an efficient \class{QMA} verifier for the problem.

An open problem related to this is to find a \class{QMA} verifier for
the \prob{Pure Fixed Point} problem, or an argument that the problem is likely to lie outside of the class.
The direct approach to construct a verifier using the swap test on (ideally) two copies of the fixed-point state, similar to the verifier in~\cite{Rosgen10non-isometry}, does not seem to work: the circuit that measures a qubit in the computational basis and then applies the Pauli $X$ gate, when applied to half of the input space, maps the symmetric state $\ket{01} + \ket{10}$ to a symmetric state.
This circuit, however, does not have any pure (approximate) fixed points.

\section*{Acknowledgements}

I am grateful for discussions with Markus Grassl, Matthew McKague, and
Lana Sheridan.  
This work has been supported by the Centre for Quantum
Technologies, which is funded by the Singapore Ministry of Education
and the Singapore National Research Foundation.

\newcommand{\arxiv}[2][quant-ph]{\href{http://arxiv.org/abs/#2}{arXiv:#2 [#1]}}
  \newcommand{\oldarxiv}[2][quant-ph]{\href{http://arxiv.org/abs/#1/#2}{arXiv:#1/#2}}

\end{document}